\documentclass[a4paper,12pt]{article}
\pdfoutput=1
\usepackage[margin=1in]{geometry}

\usepackage{amsmath}
\usepackage{physics}
\usepackage{dsfont}
\usepackage{xcolor}
\usepackage{mathrsfs}
\usepackage{graphicx}
\usepackage{authblk}

\usepackage[utf8]{inputenc}
\usepackage[backend=biber,style=alphabetic,url=false,maxbibnames=99,doi=false]{biblatex}
\newbibmacro{string+doi}[1]{
  \iffieldundef{doi}{
\iffieldundef{url}{#1}{\href{\thefield{url}}{#1}}
}{\href{http://dx.doi.org/\thefield{doi}}{#1}}}
\DeclareFieldFormat{title}{\usebibmacro{string+doi}{\mkbibemph{#1}}}
\DeclareFieldFormat[article]{title}{\usebibmacro{string+doi}{#1}}
\AtEveryBibitem{\clearfield{issn}
\clearfield{month}}
\usepackage[hidelinks]{hyperref}
\hypersetup{colorlinks=true,citecolor=blue,linkcolor=blue,filecolor=blue,urlcolor=blue,breaklinks=true}

\usepackage{tikz}
\usetikzlibrary{calc,decorations.pathreplacing}

\usepackage{amssymb}
\usepackage{amsthm}
\usepackage{pifont}
\usepackage{multirow}

\newtheorem{lemma}{Lemma}
\newtheorem{theorem}{Theorem}
\newtheorem{remark}{Remark}
\newtheorem{corollary}{Corollary}
\newtheorem{proposition}{Proposition}

\newcommand{\id}{\mathds{1}}
\newcommand{\comment}[1]{}

\newtheorem{definition}{Definition}

\newcommand{\NN}{\mathbb{N}}
\newcommand{\cA}{\mathcal{A}}

\newcommand{\CC}{\mathbb{C}}
\newcommand{\RR}{\mathbb{R}}
\newcommand{\ZZ}{\mathbb{Z}}
\newcommand{\cB}{\mathcal{B}}
\newcommand{\cH}{\mathcal{H}}
\newcommand{\cG}{\mathcal{G}}
\newcommand{\cO}{\mathcal{O}}
\newcommand{\cP}{\mathcal{P}}

\DeclareMathOperator{\diam}{diam}

\DeclareMathOperator{\corr}{Corr}

\date{}
\title{Spatial Entanglement Sudden Death in Spin Chains at All Temperatures}
\author{Samuel O. Scalet\footnote{Email: sscalet@ucdavis.edu}\\
{\normalsize Department of Computer Science, University of California, Davis, CA, 95616, USA}}

\begin{document}

\maketitle

\begin{abstract}
        We prove a finite entanglement length for the Gibbs state of any local Hamiltonian on a spin chain at any finite temperature:
        After removing an interval of size at least equal to the entanglement length, the remaining left and right half-chains are in a separable state.
\end{abstract}
\section{Introduction}

        Entanglement is one of the most fundamental concepts of quantum mechanics.
        While initially perceived as "unreal" and therefore necessarily incomplete \cite{einstein1935}, entanglement has been demonstrated as a distinguishing and verifiable feature that sets quantum mechanics apart from classical theories \cite{bell1964}. 
        Today, the entire field of quantum information theory has been built on the concept of entanglement, where it facilitates numerous applications such as quantum teleportation \cite{bennett1993}, encryption \cite{bennett1984}, and even quantum computation \cite{feynman1982}, to name just a few.

        Quantum information theory has also opened a new perspective on the field of many-body physics.
        As the underlying interactions in many-body systems are quantum, it is only natural to take into account entanglement in the study of their correlations. 
        This program has led to a plethora of new insights about connections between entanglement and correlation measures with phenomenological and computational results, see \cite{amico2008,cirac2021,alhambra2023} for some reviews on this vast area.

        \emph{Spin chains in one dimension} in their thermal state stand as a paradigmatic and also experimentally interesting example of this line of work.
        In the translation-invariant case, they have been shown to exhibit no phase transitions and to possess a unique Gibbs state by Araki \cite{araki1969,araki1975}.
        Many further results on their entanglement structure followed including the decay of other measures of correlations and entanglement, such as mutual and conditional mutual informations \cite{bluhm2022,gondolf2024,kato2019,kuwahara2024,kuwahara2022a}, efficient classical simulation \cite{kuwahara2018,kuwahara2021,fawzi2023,scalet2024,achutha2024}, sampling on quantum computers \cite{bilgin2010}, and their fast mixing for open-system dynamics \cite{bergamaschi2025}, including results that lift the assumption of translation-invariance.
        In fact, the decay of correlations in terms of correlation functions or mutual information is shown to be exponential \cite{araki1969,kimura2025,bergamaschi2025,bluhm2022}, which is qualitatively optimal as can be easily demonstrated even by a \emph{classical example} of a nearest-neighbour spin chain.
       We ask the following question:

       \begin{center}
               \textit{ What part of the exponentially decaying correlations is nonclassical?}
       \end{center}
       When asking this question, one might wonder how to quantify this entanglement.
       Typical candidates might include the entanglement of formation $E_F$, entanglement cost $E_C$, the distillable entanglement $E_D$, or distance measures to the set of separable states, including entanglement relative entropy $E_R$, which are all motivated by different operational interpretations.
       We omit the detailed definitions of all these measures and refer the reader to \cite{horodecki2009}.
       Somewhat unsatisfactorily, they can display wildly different behaviour as demonstrated by the existence of bound entangled states \cite{horodecki1998}, which have nonzero entanglement but no distillable entanglement.
       The exponential decay of mutual information bounds all of these\footnote{This is immediate for $E_R$ and the squashed entanglement, which in turn bounds the remaining ones \cite{horodecki2009,brandao2011}.}, but is not faithful, so lower bounds remain open --- in particular, it is nonzero for purely classical correlations.

       We circumvent this conundrum by proving the following stronger statement:
       \begin{center}
               \textit{The left and right ends of the chain in its Gibbs state become\\ unentangled when separated by a finite lengthscale.}
       \end{center}
       Albeit generally NP-hard to check \cite{gurvits2004}, separability has a clean definition, not suffering from the above dichotomy: A state is separable if it is a convex combination of product states, see Definition~\ref{def:sep}. 
       Stated more precisely, the theorem we prove is:
       \begin{theorem}[Informal version of Theorem~\ref{thm:mainTec}]
               For a local Hamiltonian on a spin chain, there exists a constant $\ell\in\NN$ that only depends on the temperature and locality, but not on the system size, such that the following holds.
               For any tripartite interval $ABC$ with 
               \[
                       |B|\ge\ell\,,
               \]
               and $\rho^{ABC}=\exp(-\beta H_{ABC})/Z_{ABC}$ the Gibbs state on the interval, we have that
               \[
                       \rho_{AC}=\Tr_B[\rho^{ABC}]
               \]
               is separable between $A$ and $C$.
       \end{theorem}
       This implies that all meaningful measures of entanglement, including $E_F$, $E_D$, $E_C$, and $E_{RE}$ are exactly zero on $\rho_{AC}$, and further, that the state can be prepared using only local operations and classical communication.

       We also note that this theorem holds at \emph{any fixed} temperature $1/\beta$ with no phase transition. However, as will be clear from the techniques used, the lengthscale $\ell$ grows with $\beta$. While we do not spell out the details of this dependence, the choice of techniques based on \cite{araki1969,perezgarcia2023} should result in a dependence of the form $\exp(\exp(\cO(\beta)))$.

       Since our result holds uniformly in system size, we can state the following corollary on the thermodynamic limit.
       \begin{corollary}[Informal version of Corollary~\ref{cor:KMS}]
               For a uniformly bounded finite-range interaction, consider the infinite system KMS-state $\omega$ on the quasilocal algebra and the partial trace $\tr_\Lambda$.
               There exists $\ell\in\NN$ such that
               \[
               \omega\circ\tr_{[-\ell,0]}
               \]
               is a separable state between the left and right half-chain.
       \end{corollary}
       To avoid technical details, we do not discuss this further here and refer to Section~\ref{sec:KMS} for definitions and details.

       Before proceeding to an outline of our proof, let us describe the notion of the \emph{death of entanglement} and compare it to our result.
       While not strictly defined, the term was first coined in the context of continuous-time noise channels applied to an entangled state of two atoms in cavities \cite{yu2004,yu2009}.
       Unavoidably, the correlations between the atoms decay exponentially due to the noise.
       What is somewhat more striking is that the entanglement decays not just exponentially.
       It disappears after a finite time $t_{ESD}$.
       A simplified explanation of this effect is that there is an open ball around the maximally mixed state, which consists of separable states.
       As the time-evolved state approaches the center of this ball, eventually, after a finite time, it has to lie strictly within the ball and becomes separable.

       In a totally different context, \cite{bakshi2024} recently showed that a conceptually similar phenomenon happens in high-temperature Gibbs states.
       The infinite temperature state is the maximally mixed state, and a small perturbation leaves that state separable\footnote{Their result holds in an even stronger sense, proving full multipartite separability.}.
       A caveat is that the size of the ball may depend on the dimension, and it is a highly nontrivial insight of this work to show that the temperature at which the states become separable does not grow with system size.
       Note that the high-temperature regime in which this result holds is the trivial phase, just as one-dimensional systems only possess one phase~\cite{araki1969} across all temperatures, so another natural question answered by our result is:

       \begin{center}
               \textit{Is there an analogous death of entanglement at \\any temperature for 1D Gibbs states?}
       \end{center}

       Clearly, in our setting, we cannot reproduce a result that proves separability between arbitrary pairs of spins, since simple examples of suitable interactions will yield entanglement between neighbouring spins.
       Instead, we view our result as another, novel notion of this sudden death.
       Rather than with time or temperature, we show the separability at a sufficiently large distance, which one may dub a \emph{spatial sudden death of entanglement}.

       \paragraph{Proof outline}
       We discuss the main ideas and steps of the proof of Theorem~\ref{thm:mainTec}.
       A starting point is similar to the idea already used in \cite{bakshi2024}, that for sufficiently small perturbations to the maximally mixed state, the state remains separable as long as the strength of the perturbations decays exponentially with the size of their support at a sufficiently large rate.
       This strategy is inapplicable to our case in several ways:
       \begin{itemize}
               \item We cannot tune the temperature up, so we are not close to the maximally mixed state.
       \end{itemize}
       We overcome this first barrier by combining two results: the exponential decay of correlations in the form of an \emph{approximate factorization}
       \begin{equation}\label{eq:appFac}
               \rho_{AC}\approx \rho_A\otimes\rho_C
       \end{equation}
and the \emph{quantitative faithfulness} of marginals, i.e., lower bounds on their smallest eigenvalues $\rho_{A/C}\succeq \alpha\id>0$.
In particular, the latter implies that, albeit far from the maximally mixed state, the marginals possess a positive decomposition with a maximally mixed component, which remains separable when adding the deviation from the product state due to the only approximate Equation~\eqref{eq:appFac}.
       This result forms Proposition~\ref{prop:constACsep}, which already implies the separability for constant sizes of the regions $A$ and $C$.
       Interestingly, 1D is not a crucial assumption for this statement as long as uniform faithfulness and decay of correlations can be proven.

       However, the above argument is insufficient to extend to arbitrary system sizes for two reasons:
       \begin{itemize}
               \item The exponential decay of the approximate factorization is optimal, and the rate goes to  $0$ as $\beta$ increases.
       \end{itemize}
       Again, this is true due to classical correlations.
       A calculation using transfer matrices for a simple classical nearest-neighbour Ising chain yields the optimal exponential decay of correlations as well as mutual information, which lower bounds the error of the approximate factorization. Moreover:

       \begin{itemize}
               \item The exponentially small faithfulness $\alpha$ is optimal for a rate $\log(d)+\Omega(\beta)$.
       \end{itemize}
       This lower bound is already satisfied for a noninteracting spin chain.

       To overcome this issue, we show an additional quasilocal substructure of the deviation from the product state $\rho_{AC}-\rho_A\otimes\rho_C$: 
       It splits into a term with small variable support size $k_0$, which is exponentially decaying in $|B|$, and an additional tail of terms with larger supports $\Delta_k$, $k\ge k_0$.
       For the former part, we also find a decomposition into a separable and a maximally mixed component $\Gamma+\gamma\id$.
       Only for the additional tail, we are then able to show an even stronger, superexponential decay in their support size and $|B|$.
       For any given temperature, by increasing the support of the first part, this decay rate of $\Delta_k$ can now exceed the decay rate of $\gamma$.
       Viewing the tail as a perturbation of the maximally mixed component, their sum $\gamma\id+\sum_k\Delta_k$ remains separable.
       By combining all these separable contributions, the theorem follows.

       \paragraph{Related work}
       The question of the decay of bipartite entanglement has been studied in \cite{kuwahara2022a,kuwahara2024}. In particular, the latter shows the decay of the squashed entanglement in one and higher dimensions, though the higher-dimensional case is restricted to constant-sized subsystems.
       In fact, due to the lifted translation-invariance assumption in the decay of correlations \cite{bergamaschi2025} in one dimension, this is by now subsumed by the decay of mutual information \cite{bluhm2022}, see Lemma~\ref{lem:doC}.
       However, while addressing meaningful measures of entanglement, all of these prior bounds prove only decay and not the stronger sudden death of entanglement.

       While preparing this manuscript, we became aware of \cite{bakshi2026}, which proves a strictly finite Schmidt rank for finite-range thermal states in one dimension.
       While the concept of proving a strictly finite property of the entanglement, where generically one might only expect a decaying spectrum, seems related, there is no technical connection between the statements in their work and ours.
       While the Schmidt rank bounds the amount of entanglement between neighbouring regions, our result concerns the decay with distance.
       In fact, the authors mention the open question of the entanglement lengthscale, which our work answers.

       \paragraph{Discussion and outlook}
       In this paper, we prove the sudden death of entanglement in one-dimensional Gibbs states.
       Tracing out a finite subsystem from a chain, the state becomes \emph{exactly} separable between the two halves, despite inevitable quantum correlations between close neighbours.
       To the best of our knowledge this is the first nontrivial example of such a \emph{spatial} death of entanglement.
       
       There are a number of interesting open questions:
       While in our setting of one-dimensional Gibbs states, classical and quantum correlations are extensively classified, it remains an open question whether building on our techniques a superexponential decay of the conditional mutual information can be shown.
       While its exponential decay rate is proven in \cite{kuwahara2024}, no lower bound is known and in particular, in the classical setting the stronger Hammersley-Clifford theorem implies exact Markovianity beyond the interaction range.

       Further, regarding the death of entanglement, it would be interesting to prove analogous results in higher-dimensional low-temperature systems or ground states.
       While this clearly goes beyond the scope of our techniques, features like the monogamy of entanglement that rule out arbitrarily strong entanglement between multiple particles might give a hint in this direction.
       For ground states, while their marginals are low-rank, some perturbations of low-rank separable states can still remain separable.
       An indication, that the ground state regime is not hopeless is given by the following trivial example:
       States prepared by constant-depth quantum circuits exhibit the death of entanglement, but maintain short-range entanglement.

       \paragraph{}\emph{Outline} In Section~\ref{sec:prelim}, we introduce the notation and prior and elementary results needed for the proof, which follows in Section~\ref{sec:proof} split into the constant-size result Proposition~\ref{prop:constACsep} and the main Theorem~\ref{thm:mainTec}.
       In Section~\ref{sec:KMS}, we discuss the thermodynamic limit and prove Corollary~\ref{cor:KMS}.

       \section{Preliminaries}\label{sec:prelim}
In this section, we introduce the setup and notation as well as some elementary or known results needed for the proof.

For a finite-dimensional Hilbert space $\cH$, we denote the set of bounded operators $\cB(\cH)$. 
A quantum state is a normalized positive semidefinite operator $\rho\in\cB(\cH)$, $\rho\succeq0$, $\Tr[\rho]=1$, where $\succeq$ denotes the positive semidefinite  order.
$\Tr[\cdot]$ denotes the trace, $\|\cdot\|$ the operator norm, and $\|X\|_1=\Tr[\sqrt{X^*X}]$ the trace norm.
For two positive quantum states, the Umegaki relative entropy is defined as $D(\rho\|\sigma)=\Tr[\rho(\log(\rho)-\log(\sigma))]$.
\subsection{Separable states}

\begin{definition} \label{def:sep}
        A positive semidefinite matrix $X\in \cB(\CC^{d_A})\otimes\cB(\CC^{d_C})$ is called separable if there exists a decomposition
        \[
                X=\sum_{i=1}^r X_i^A\otimes X_i^C\,,
        \]
        where $X_i^A\in\cB(\CC^{d_A})$, $X_i^C\in\cB(\CC^{d_C})$ and $X_i^A,X_i^C\succeq0$.
        The set of separable states in finite dimensions is the set of separable matrices $X$ with $\Tr[X]=1$.
        A matrix or state that is not separable is called entangled.
\end{definition}
A famous criterion for entanglement is the PPT criterion: From the above decomposition, it is evident that transposing one subsystem of a separable state results in a (separable) state, where the linear transposition map is $(\cdot)^{T_B}:\ket i\bra j\otimes\ket k\bra l\mapsto\ket i\bra j\otimes\ket l\bra k$.
However, for some states $\rho^{T_B}$ is not positive semidefinite, which, by contradiction, proves that they are entangled. 
The criterion, however, is not faithful except for dimensions $2\times 2$ and $2\times 3$ \cite{woronowicz1976}.
While characterization of bipartite pure state entanglement is well understood, the landscape of mixed state entanglement measures is fragmented and depends on the specific application as reviewed in \cite{horodecki2009}.
As our proof of exact separability works in an elementary way with Definition~\ref{lem:sepStruct}, we do not go into detail.

We collect a number of elementary properties of separable operators in finite dimensions, which we will use repeatedly.
\begin{lemma}\label{lem:sepStruct}
        The set of separable matrices is a closed cone, i.e., it is closed under linear combinations with positive coefficients.
        The set of separable states is a closed conic section.
        Further, if $X_{AC}\in\cB(\CC^{d_A})\otimes\cB(\CC^{d_C})$ is separable, so is
        \[
                (Y_A\otimes Y_C)X(Y_A\otimes Y_C)^\dagger
        \]
        for any $Y_A\in\cB(\CC^{d_A})$, $Y_C\in\cB(\CC^{d_C})$.
\end{lemma}
\begin{proof}
        The cone property is immediate from the definition.
        To show that it is closed, consider a sequence of separable states $X_k$ that converges to $X\in\cB(\CC^{d_A})\otimes \cB(\CC^{d_C})$.
        By Carath\'eodory's theorem, we can always choose $r=(d_Ad_B)^2$ in Definition~\ref{def:sep}, so we can write
        \[
                X_k=\sum_{i=1}^r p_{k,i} (X_k)_i^A\otimes (X_k)_i^C
        \]
        with $\Tr[(X_k)_i^A]=\Tr[(X_k)_i^C]=1$ and $p_{k,i}\ge0$. Since $X_k$ converges, it is bounded, and so
        $p_{k,i}=\left\|p_{k,i} (X_k)_i^A\otimes (X_k)_i^C\right\|_1\le\|X_k\|_1$ is bounded too.
        Now, all of $p_{k,i}$, $(X_k)_i^A$, and $(X_k)_i^C$ are bounded sequences in finite dimensions, and hence we can descend to a subsequence $k_l$, for which all of them converge, and
        \[
                X=\lim_{k\to\infty} X_k=\lim_{l\to\infty} \sum_{i=1}^r p_{k_l,i} (X_{k_l})_i^A\otimes (X_{k_l})_i^C=\sum_{i=1}^r p_i X_i^A\otimes X_i^C
        \]
        is separable.
        The set of separable states is the intersection of the separable matrices with the closed affine space defined by $\Tr[x]=1$, and therefore a closed conic section.
        Since $X\succeq0$ implies $YXY^\dagger\succeq0$ for any matrix $Y$, the last part follows from the definition of a separable matrix.
\end{proof}
Since the set of separable matrices and states are convex and thereby (path-)connected and closed, they are not open (in the subspace topology).
An explicit example is as follows:
The sequence of states $(1-1/n)\ketbra{00}{00}+(\ket{00}+\ket{11})(\bra{00}+\bra{11})/2n$ converges to the separable state $\ketbra{00}{00}$ but for each $n$ does not satisfy the PPT criterion and is therefore entangled.

However, there exists an open ball around the identity in which all states are separable as stated in the following Lemma, which is central to our main theorem.
The optimal dimension dependent size of such balls in various norms has been studied in \cite{Gurvits_2002} and we adopt the following result. 
\begin{lemma}[{\cite[Theorem 2]{Gurvits_2002}}]\label{lem:sepPerturb}
        For $\Delta\in\cB(\CC^{d_A})\otimes\cB(\CC^{d_C})$ with $\|\Delta\|\le 1/\sqrt{d_Ad_C}$ 
        \[
                \id_{AC}+\Delta
        \]
        is separable between $A$ and $C$.
\end{lemma}
Since the above statement trivially continues to hold when tensoring with the identity, the dimension can be restricted to  the dimension of the \emph{support} of the operator $\Delta$ in a multipartite Hilbert space.
Note that in \cite{bakshi2024} a version of this statement for qubits even proves multipartite separability.
In the context of our work nearest-neighbour entanglement is naturally unavoidable, so we restrict ourselves to the above bipartite statement with slightly better dimension dependence.
\subsection{1D Gibbs states}
We consider a spin system on the chain $\ZZ$ and consider the local algebra $\cA_0=\cup_{n\in\NN} \cA_{[-n,n]}$, where $\cA_{X}=\otimes_{i\in X}\cA_i$, $\cA_i\cong\cB(\CC^d)$ and up to the natural embedding, i.e., for $\Lambda\subset\Lambda'$, $X_\Lambda\in\cA_\Lambda$ is identified with $X_\Lambda\otimes \id_{\Lambda'\setminus\Lambda}$.
Its closure $\cA=\overline{\cA_0}$ in operator norm is the quasilocal algebra.
$\id$ denotes the identity matrix.
We denote the trace on a region $\Lambda$, by $\Tr_\Lambda:\cA_\Lambda\to \CC$ and omit the subscript when clear from the context.
For $X\subset\Lambda$, $\tr_X:\cA_\Lambda\to\cA_{\Lambda\setminus X}$ denotes the (unnormalized) partial trace, i.e., $\tr_X[\id_\Lambda]=d^{|X|}\id_{\Lambda\setminus X}$.
For an operator $X\in\cA_0$, its support is the smallest set $\Lambda$ such that $X\in\cA_\Lambda$.

For $\cP_0(\ZZ)$ the set of finite subsets of $\ZZ$, let $\Phi:\cP_0(\ZZ)\to\cA_0$ be a local interaction, i.e., such that $\Phi(X)\in\cA_X$.
We assume that it has a finite-range $r$ such that $\Phi(X)=0$ if $\diam(X)>r$ and bounded strength $J$
\[
        J=\sup_{i\in\ZZ} \sum_{X:\;i\in X} \|\Phi(X)\|
\]
Throughout the paper, when writing that a term is constant, we mean that it is an explicit function only of $J$, $r$, and $d$, but not any system size.
For any region $A$, we define
\[
        H_A=\sum_{X\subset A}\Phi(X)
\]
\begin{figure}[ht!]
        \centering
        \begin{tikzpicture}[
    site/.style={circle, draw, fill=black, minimum size=3mm, inner sep=0pt},
    every node/.style={font=\small}
]

\def\nA{4}
\def\nB{5}
\def\nC{4}
\def\k{2}

\pgfmathtruncatemacro{\N}{\nA+\nB+\nC}

\foreach \i in {1,...,\N} {
    \node[site] (s\i) at (\i,0) {};
}

\pgfmathtruncatemacro{\startB}{\nA+1}
\pgfmathtruncatemacro{\endB}{\nA+\nB}
\pgfmathtruncatemacro{\startC}{\endB+1}

\pgfmathtruncatemacro{\startPk}{\startB-\k}
\pgfmathtruncatemacro{\endPk}{\endB+\k}

\draw[dashed]
    ($(s\nA.east)!0.5!(s\startB.west)+(0,-.5)$) -- ++(0,1.0);
\draw[dashed]
    ($(s\endB.east)!0.5!(s\startC.west)+(0,-.5)$) -- ++(0,1.0);


\draw[decorate,decoration={brace,mirror,amplitude=4pt}]
    ($(s1.south)+(-0.2,-0.4)$) -- ($(s\nA.south)+(0.2,-0.4)$)
    node[midway,below=5pt] {$A$};

\draw[decorate,decoration={brace,mirror,amplitude=4pt}]
    ($(s\startB.south)+(-.2,-0.4)$) -- ($(s\endB.south)+(0.2,-0.4)$)
    node[midway,below=5pt] {$B$};

\draw[decorate,decoration={brace,mirror,amplitude=4pt}]
    ($(s\startC.south)+(-.2,-0.4)$) -- ($(s\N.south)+(.2,-0.4)$)
    node[midway,below=5pt] {$C$};


\draw[decorate,decoration={brace,amplitude=4pt}]
    ($(s\startPk.north)+(-.2,0.6)$) -- ($(s\endPk.north)+(0.2,0.6)$)
    node[midway,above=5pt] {$\partial^k B$};

\end{tikzpicture}
        \caption{\label{fig:regions}Graphical representation of the regions $ABC$.}
\end{figure}
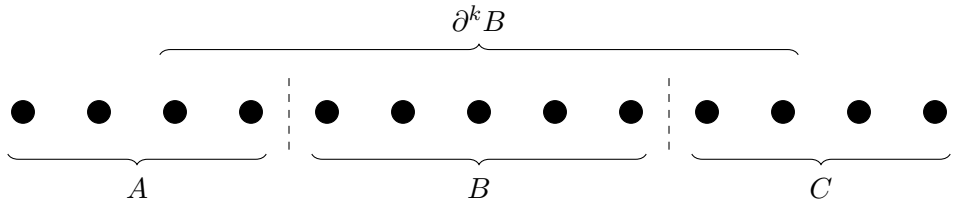
We consider a finite interval $ABC\subset \ZZ$ with $A$, $B$, and $C$ intervals with $A$ preceding and $C$ succeeding $B$.
We use the following notation for the $k$-neighbourhood of $B$: $\partial^k B=\{k\in\Lambda|d(k,B)\le k\}$, see Figure~\ref{fig:regions}.

We denote the Gibbs state on a region $\Lambda$ by $\rho^\Lambda=\exp(-H_\Lambda)/Z_\Lambda$, $Z_\Lambda=\Tr_\Lambda[\exp(-H_\Lambda)]$
and its marginals $\rho_X^\Lambda=\tr_{\Lambda\setminus X}[\rho^\Lambda]$, where we omit the superscript when clear from the context.
Here, we choose the inverse temperature $\beta=1$ as it can be absorbed into the interaction strength $J$.
Note that $\rho^\Lambda$ will sometimes be considered as an operator in $\cA_0$ but it is only normalized in $\cA_\Lambda$.
If not noted otherwise $\rho=\rho^{ABC}$.
We use the shorthand $\rho^k=\rho^{\partial^k B}$, $H^k_X=H_{X\cap\partial^k B}$, and $Z_{k_0}=Z_{\partial^k B}$.

We need the following standard Lemma on the growth of partition functions.
\begin{lemma}\label{lem:partFuncRatio}
        For Hamiltonians $H,H'\in\cB(\CC^d)$,
        \begin{equation}\label{eq:Z1}
        \frac{\Tr[\exp(-H')]}{\Tr[\exp(-H)]}\le\exp\left(\left\|H-H'\right\|\right).
\end{equation}
In particular, for an interaction of strength $J$ and $Z_\Lambda$ as defined above and any region $A$
\begin{equation}\label{eq:Z2}
        \exp((\log(d)-J)|A|)\le Z_A\le\exp((\log(d)+J)|A|)
\end{equation}
and for neighbouring intervals $A$, $B$,
\begin{equation}\label{eq:Z3}
        \exp(-rJ)\le\frac{Z_AZ_B}{Z_{AB}}\le\exp(rJ)\,.
\end{equation}
\end{lemma}
\begin{proof}
        We introduce the smooth function $f:[0,1]\to\RR$
        \[
                f(s)=\Tr[e^{-H+s(H-H')}].
        \]
        By Duhamel's formula
        \[
                f'(s)=\Tr[(H-H')e^{-H+s(H-H')}]\le\|H-H'\| f(s)
        \]
        and so by Gr\"onwall's inequality $f(1)\le f(0)\exp(\|H-H'\|)$, which proves \eqref{eq:Z1}.
        Equation~\eqref{eq:Z2} follows by choosing $H=H_A$, $H'=0$ and vice versa.
        Equation~\eqref{eq:Z3} follows by choosing $H=H_A+H_B$, $H'=H_{AB}$ and vice versa and noting that the terms in $H_{AB}$ that differ from the ones in $H_A+H_B$ have support containing one of the first $r$ sites of $B$, so $\|H_A+H_B-H_{AB}\|\le rJ$.

\end{proof}

In \cite{araki1969} the Araki expansionals were introduced and their locality structure analyzed.
The result allows to relate an (unnormalized) Gibbs state to a product of Gibbs states on decoupled subintervals and is based on Lieb-Robinson type bounds for the complex time evolution.
We use the following alternative formulation of Araki's results from a more recent work (see also \cite{perezgarcia2023}).
\begin{lemma}[{\cite[Corollary 3.4]{bluhm2022}}]\label{lem:arakiExp}
Let $\Phi$ be a local interaction of range $r$ and strength $J$, $s\in\CC$ with $|s|\le1$, and consider 
\[
        E_{X,Y}(s):=e^{-sH_{XY}}e^{sH_X+sH_Y} \in \cA_{X\cup Y}
\]
where $X,Y$ are two adjacent intervals.
Then there exists a constant $\cG>0$ only dependent on $r$ and $J$ but not on the size of $X$, $Y$ such that
\[
        \|E_{X,Y}(s)\|,\|E_{X,Y}(s)^{-1}\|\le\cG
\]
uniformly in the size of $X$ and $Y$.
Further, let $\tilde X$, $\tilde Y$ be (possibly empty) intervals immediately preceeding and succeeding $X$ and $Y$ respectively, and $|X|,|Y|\ge\ell$. Then for $|s|\le1$
\[
        \left\|E_{X,Y}(s)-E_{\tilde XX,Y\tilde Y}(s)\right\|,\left\|E_{X,Y}(s)^{-1}-E_{\tilde XX,Y\tilde Y}(s)^{-1}\right\|\le\frac{\cG^\ell}{(\lfloor\ell/r\rfloor+1)!}\,.
\]

\end{lemma}
As another shorthand notation we define the Araki expansional restricted to the $k$-neighbourhood of $B$:
\[
				E^k_{X,Y}(s)=E_{X\cap\partial^k B,Y\cap\partial^k B}
\]

We will make use of results on the decay of correlations in Gibbs states.
In particular, we need a form of approximate factorization to prove separability of constant-sized regions.
The results in \cite{bluhm2022} show the equivalence of several notions of the decay of correlations including that of correlation functions as well as said approximate factorization.
While the final results in there are restricted to the translation-invariant case, as the result relies on \cite{araki1969}, this restriction can be readily lifted by applying the equivalence to the more recent result on decay of correlations in \cite{bergamaschi2025}.

\begin{lemma}\label{lem:doC}
        For any consecutive adjacent intervals $\tilde AABC\tilde C$, and potential $\Phi$ of bounded range $r$ and strength $J$, consider the Gibbs state $\rho=\rho^{\tilde AABC\tilde C}$. Then, there exist constants $C, \alpha>0$ only dependent on $r$, $J$, and $d$ such that
 \[
         \|\rho_{AC}-\rho_A\otimes\rho_C\|\le C \exp(-\alpha|B|)
 \]
\end{lemma}
\begin{proof}
        \cite[Corollary I.2]{bergamaschi2025} proves the exponential decay of correlations
        \begin{align*}
                \left|\corr(X^A,X^C)_{\rho_{AC}}\right|:&= \left|\Tr_{AC}\left[(X^A\otimes X^C)\rho_{AC}\right]-\Tr_A[X^A\rho_A]\Tr_C[X^C\rho_C]\right|\\
                &\le C'\exp(-\alpha'|B|)
        \end{align*}
        for constants $C',\alpha'>0$ only dependent on the interaction range, strength, and local dimension.
        This is equivalent to the uniform exponential clustering condition in \cite[Theorem 8.2]{bluhm2022}, which thereby yields the exponential decay of the mutual information
        \[
        I(A:C)\le C''\exp(-\alpha''|B|)\,.
        \]
        The Lemma follows from Pinsker's inequality $\|\rho_{AC}-\rho_A\otimes\rho_C\|_1^2\le2 I(A:C)$ and the matrix norm inequality $\|\cdot\|\le\|\cdot\|_1$.
\end{proof}

\begin{remark}\label{rem:doC}
        The equivalence in \cite{bluhm2022} is inevitably restricted to 1D systems making heavy use of the results in Lemma~\ref{lem:arakiExp}.
        However, it turns out that a weaker version of the equivalence between decay of correlation functions and mutual information with \emph{an exponential overhead in the sizes} of $A$ and $C$ holds more generally.
        It is a straightforward consequence of the equivalence between trace norm and  relative entropy distance together with the equivalences  of finite-dimensional norms up to dimensional prefactors.
       It turns out that this exponential overhead poses no barrier to our proof.
       The above lemma makes our proof cleaner and improves constants.
       However, for potential extensions to higher dimensions, it is interesting to note that the decay of correlations is sufficient to ensure this weaker form of approximate factorization.
\end{remark}

The following Lemma proves that the decay of smallest eigenvalues of marginals is no faster than exponential.
Due to a slightly different setup, we present the proof which is analogous to \cite[Lemma 4.2]{fawzi2023}.
\begin{lemma}\label{lem:margLowerBound}
        For an interaction $\Phi$ of range $r$ and strength $J$, there are constants $C$, $\alpha>0$, only dependent on $r$, $J$, and $d$, such that for any tripartite interval $ABC$, and the Gibbs state $\rho^{ABC}$ its marginal is lower bounded as
        \[
               \|\rho_B^{-1}\|\le C\exp(\alpha|B|)
        \]
        uniformly in the size of $A$ and $C$.
\end{lemma}
\begin{proof}
        Consider
\[
        E_{AB,C}(-1/2)E_{A,B}(-1/2)E_{A,B}^\dagger(-1/2)E_{AB,C}^\dagger(-1/2)=e^{H_{ABC}/2}e^{-(H_A+H_B+H_C)}e^{H_{ABC}/2}\,.
\]
By Lemma~\ref{lem:arakiExp}, we have
\[
        \left\|E_{AB,C}(-1/2)E_{A,B}(-1/2)E_{A,B}^\dagger(-1/2)E_{AB,C}^\dagger(-1/2)\right\|\le\cG^4
\]
or equivalently
\[
        \cG^{4}\id\succeq e^{H_{ABC}/2}e^{-(H_A+H_B+H_C)}e^{H_{ABC}/2}\,.
\]
We rearrange
\[
        \cG^{4}e^{-H_{ABC}}\succeq e^{-(H_A+H_B+H_C)}\,,
\]
and using that partial traces preserve the positive semidefinite order 
\[
        Z_{ABC}\cG^{4}\rho_B\succeq Z_AZ_C \exp(-H_B)\succeq Z_AZ_C \exp(-\|H_B\|)\id\,.
\]
Since, by Lemma~\ref{lem:partFuncRatio}
\[
        \frac{Z_AZ_C}{Z_{ABC}}=\frac{Z_{AB}Z_C}{Z_{ABC}}\frac{Z_AZ_B}{Z_{AB}}\frac1{Z_B}\ge \exp(-(J+\log(d))|B|-2rJ)\,,
\]
and $\|H_B\|\le J|B|$, the proof follows with $C=\cG^{4}e^{2rJ}$ and $\alpha=2J+\log(d)$.
\end{proof}
\begin{remark}
        A simple example of a non-interacting, translation-invariant on-site potential shows that the above Lemma is optimal with $\alpha\ge\log(d)+J$.
\end{remark}

\section{Proof of Theorem~\ref{thm:mainTec}}\label{sec:proof}
In this section, we present the proof of the main theorem.
We will make use of the following technical lemma.
\begin{lemma}\label{lem:pTrContract}
        For any regions $A$, $B$, and a normalized state $\rho_B\in\cA_B$, the map $\cA_{AB}\to\cA_A$, $X\mapsto\tr_B[\rho_B X]$ is contractive in the operator norm.
\end{lemma}
\begin{proof}
        Let $\rho_B=\sum_i p(i) \ketbra{i}{i}_B$ where $\ket i_B$ is an orthonormal basis diagonalizing $\rho_B$. Then, expanding the partial trace in this basis, we have
        \begin{align*}
                \left\|\tr_B[\rho_B X]\right\|&=\sup_{\ket{\psi}_A: \braket\psi_A=1} \sum_ip(i)\bra{\psi}_A\bra{i}_B X \ket{\psi}_A\ket i_B\\
                &\le\sum_i p(i) \|X\|\\
                &=\|X\|\,.
        \end{align*}
\end{proof}

Next, we prove a separable structure of the marginals $AC$, where the entanglement lengthscale $|B|$ still depends on the size of $A$ and $C$.
We will bootstrap the general, system-size independent entanglement lengthscale from this proposition.
\begin{proposition}\label{prop:constACsep}
        Let $\Phi$ be an interaction of range $r$ and strength $J$, and fix the local dimension $d$. There exists a function $\ell_1(k)$ (also dependent on $r$, $J$, and $d$) such that for $|B|\ge\ell_1(k)$, there exists a decomposition
\[
        e^{H^k_{AC}/2}\rho_{AC}^{k}e^{H^k_{AC}/2}=\gamma(k)\id+\Gamma(k)\,,
\]
where $\gamma(k)=\exp(-\mathcal O(k))\id$ and $\Gamma(k)$ is separable.
Furthermore, we can choose $\ell_1(k)=\cO(k)$.
\end{proposition}
\begin{proof}
        We assume $|B|\ge r$, such that $H^k_{AC}=H_A^k+H_C^k$.
        Define $\Delta$ by 
        \[
                \tilde\rho_{AC}^k:=e^{H^k_{AC}/2}\rho_{AC}^ke^{H^k_{AC}/2}=\tilde\rho_A^k\otimes \tilde\rho_C^k +\Delta\,,
        \]
where $\tilde\rho_A^k:=e^{H^k_A/2}\rho_A^ke^{H^k_A/2}$ and $\tilde\rho^k_C:=e^{H_C^k/2}\rho_C^ke^{H_C^k/2}$.
By Lemma~\ref{lem:doC}, there exist constants $C, \alpha>0$ such that 
\begin{align*}
        \|\Delta\|&=\|e^{H_{AC}^k/2}\rho^k_{AC}e^{H^k_{AC}/2}-\tilde\rho^k_A\otimes\tilde\rho^k_C\|\\
                  &\le \left\|e^{H^k_{AC}}\right\| \left\|\rho^k_{AC}-\rho^k_A\otimes\rho_C^k\right\|\\
                &\le C \exp(2Jk-\alpha|B|)\,,
        \end{align*}
where we use $\left\|\exp(H^k_{AC})\right\|\le\exp(2Jk)$.

Furthermore, using Lemma~\ref{lem:margLowerBound}, we have that
\[
        \tilde\rho^k_A=e^{H_A^k/2}\rho_A^ke^{H_A^k/2}\succeq e^{H_A^k}\left\|(\rho^k_A)^{-1}\right\|^{-1}\succeq \left\|e^{-H_{A}^k}\right\|^{-1}\left\|(\rho^k_A)^{-1}\right\|^{-1}\succeq C'\exp(-\alpha'k)\id\,,
\]
where we combined the contribution from the bound $\|\exp(-H_A^k)\|\le\exp(Jk)$ and the constant from the Lemma into a new constant $\alpha'$. 
The same bound holds for $\tilde\rho^k_C$.

We choose $\gamma(k)=C'^2\exp(-2\alpha' k)/2$ and decompose
\begin{align*}
        \tilde\rho_{AC}^k&=\tilde\rho_A^k\otimes\tilde\rho_C^k+\Delta\\
                   &=\left(\tilde\rho_A^k-\sqrt{2\gamma(k)}\id_A\right)\otimes\left(\tilde\rho_C^k-\sqrt{2\gamma(k)}\id_C\right)+\sqrt{2\gamma(k)}\id_A\otimes\left(\tilde\rho_C^k-\sqrt{2\gamma(k)}\id_C\right)\\
                   &\phantom{=}+\left(\tilde\rho_A^k-\sqrt{2\gamma(k)}\id_A\right)\otimes\sqrt{2\gamma(k)}\id_C+2\gamma(k)\id_{AC}+\Delta\,.
\end{align*}
Consequently,
\begin{align*}
        \Gamma(k):=&\left(\tilde\rho_A^k-\sqrt{2\gamma(k)}\id_A\right)\otimes\left(\tilde\rho_C^k-\sqrt{2\gamma(k)}\id_C\right)+\sqrt{2\gamma(k)}\id_A\otimes\left(\tilde\rho_C^k-\sqrt{2\gamma(k)}\id_C\right)\\
                   &\phantom{=}+\left(\tilde\rho_A^k-\sqrt{2\gamma(k)}\right)\otimes\sqrt{2\gamma(k)}+\gamma(k)\id_{AC}+\Delta
\end{align*}
and due to our choice of $\gamma$, the first three terms are separable as
\begin{align*}
        \tilde\rho^k_A-\sqrt{2\gamma(k)}\id_A,\,\tilde\rho^k_C-\sqrt{2\gamma(k)}\succeq 0
\end{align*}
are positive operators.
Finally, by Lemma~\ref{lem:sepPerturb} the remaining term
\[
        \gamma(k)\id_{AC}+\Delta
\]
is separable if 
\[
\frac{\|\Delta\|}{\gamma(k)}=\frac{C\exp(2Jk-\alpha|B|)}{C'^2\exp(-2\alpha'k)/2}\le d^{-k}\,,
\]
which is achieved by the choice $|B|\ge\ell_1=\max\{r,(\log(2C/C'^2)+k(2\alpha'+\log(d)+2J))/\alpha\}=\cO(k)$.
\end{proof}
\begin{remark}
        In fact, the above theorem already proves separability of \emph{constant-sized} regions (the factors $e^{H^k_{AC}/2}$ in the statement are superfluous by Lemma~\ref{lem:sepStruct} and included only for convenience later).
        Furthermore, it only relies on two ingredients:
        Uniform faithfulness of the state, i.e., a uniform lower bound on all constant-sized marginals, and decay of correlations.
        For the latter, even a weaker notion than the approximate factorization suffices since for constant dimension they can be proven equivalent, see Remark~\ref{rem:doC}.
        This paves the way for higher-dimensional extensions of the argument, though it would only be interesting for decay of correlations beyond the high-temperature regime as otherwise, the result of \cite{bakshi2024} already applies.
\end{remark}

Let us now proceed to the proof of our main theorem.

\begin{theorem}\label{thm:mainTec}
		Let $\Phi$ be a potential on a one-dimensional spin-chain of range $r$ and interaction strength $J$. There is a function $\ell(d,J,r)$, such that for any interval $A\cup B\cup C$ with disjoint contiguous subintervals $A$, $B$, and $C$ in order, where $|B|\ge \ell(d,J,r)$ and the Gibbs state $\rho=\rho^{ABC}=\exp(-H_{ABC})/Z_{ABC}$, its marginal $\rho_{AC}$ is separable between $A$ and $C$. In particular, $\ell(d, J, r)$ does not depend on $|A|$ or $|C|$.
\end{theorem}
\begin{remark}
        While stated for adjacent intervals $ABC$, which cover the entire system, it follows immediately that separability holds for any subregions at distance $\ell+1$ as the set of separable states is invariant under LOCC protocols including partial traces.
\end{remark}
\begin{proof}
        We assume again that $|B|\ge r$ such that $H_{AC}=H_A+H_C$.
        By Lemma~\ref{lem:sepStruct}, to show separability of $\rho_{AC}$, we can equivalently show that the following operator is separable.
\begin{align}
        \frac{Z_{ABC}}{Z_B}&\exp(H_{AC}/2)\rho_{AC}\exp(H_{AC}/2)=\frac1{Z_B}e^{H_{AC}/2}\tr_B\left[e^{-H_{ABC}}\right]e^{H_{AC}/2}\label{eq:sandwMarg}\\
											  &=\tr_B\Big[\rho^B e^{(H_A+H_B)/2}e^{-H_{AB}/2}e^{H_{AB}/2}e^{H_C/2}e^{-H_{ABC}}\\
                        &\phantom{=}\,\qquad\times e^{H_C/2}e^{H_{AB}/2}e^{-H_{AB}/2}e^{(H_A+H_B)/2}\Big]\\
                        &=\tr_B\Big[\rho^B E_{A,B}^\dagger(1/2)E_{AB,C}^\dagger(1/2)E_{AB,C}(1/2)E_{A,B}(1/2)\Big]\\
                        &=\tr_B\Big[\rho^B E^{k_0\dagger}_{A,B}(1/2)E^{k_0\dagger}_{AB,C}(1/2)E^{k_0}_{AB,C}(1/2)E^{k_0}_{A,B}(1/2)\Big]\label{eq:k0region}\\
												&\phantom{==}+\sum_{k=k_0}^\infty\Big(\tr_B\Big[\rho^B E^{k+1\dagger}_{A,B}(1/2)E^{k+1\dagger}_{AB,C}(1/2)E^{k+1}_{AB,C}(1/2)E^{k+1}_{A,B}(1/2)\Big]\\
                        &\phantom{==+\sum_{k=k_0}^\infty\Big(}-\tr_B\Big[\rho^B E^{k\dagger}_{A,B}(1/2)E^{k\dagger}_{AB,C}(1/2)E^{k}_{AB,C}(1/2)E^{k}_{A,B}(1/2)\Big]\Big)\label{eq:decompLast}
\end{align}
\begin{figure}
        \centering
        \begin{tikzpicture}[
        scale=0.7, transform shape=false,
    spin/.style={circle,draw,fill=black,minimum size=2mm,inner sep=0pt},
    iden/.style={font=\small},
    every node/.style={font=\small}
]

\def\nA{7}
\def\nB{5}
\def\nC{7}
\pgfmathtruncatemacro{\N}{\nA+\nB+\nC}

\newcommand{\drawfullchain}[1]{
    \foreach \i in {1,...,\N}{
        \node[spin] at (\i,#1) {};
    }
}

\newcommand{\drawACline}[3]{

\pgfmathtruncatemacro{\startB}{\nA+1}
\pgfmathtruncatemacro{\endB}{\nA+\nB}
\pgfmathtruncatemacro{\Acut}{\nA-#2}   

\foreach \i in {1,...,\nA}{
    \ifnum\i>\Acut
        \node[spin,opacity=#3] at (\i,#1) {};
    \else
        \node at (\i,#1) {$\mathds{1}$};
    \fi
}


\foreach \i in {1,...,\nC}{
    \pgfmathtruncatemacro{\pos}{\endB+\i}
    \pgfmathtruncatemacro{\nplus}{#2+1}
    \ifnum\i<\nplus
        \node[spin,opacity=#3] at (\pos,#1) {};
    \else
        \node at (\pos,#1) {$\mathds{1}$};
    \fi
}
}


\drawfullchain{0}

\pgfmathtruncatemacro{\midA}{\nA/2+1}
\pgfmathtruncatemacro{\midB}{\nA+\nB/2+1}
\pgfmathtruncatemacro{\midC}{\nA+\nB+\nC/2+1}

\node at (\midA,0.9) {$A$};
\node at (\midB,0.9) {$B$};
\node at (\midC,0.9) {$C$};


\drawACline{-1.8}{3}{1}

        \node[right] at (\N+1,-1.8) {$\Gamma(3)$};

\node at (\midB,-1.8){separable};


\drawACline{-3.4}{4}{1}

\node[right] at (\N+1,-3.4) {$\Delta_{3}$};


\drawACline{-5.0}{5}{0.6}

\node[right] at (\N+1,-5.0) {$\Delta_{4}$};


\drawACline{-6.6}{6}{0.3}

\node[right] at (\N+1,-6.6) {$\Delta_{5}$};

\pgfmathsetmacro{\xAB}{\nA + 0.5}
\pgfmathsetmacro{\xBC}{\nA + \nB + 0.5}

\def\ytop{0.9}
\def\ybottom{-7.2}

\draw[dotted] (\xAB,\ytop) -- (\xAB,\ybottom);
\draw[dotted] (\xBC,\ytop) -- (\xBC,\ybottom);
\draw[dotted] (1,-2.6) -- (\N,-2.6);
\node at (\midB,-4.2) {entangled,};
\node at (\midB,-5.2) {decaying};

\end{tikzpicture}
        \caption{\label{fig:decomp} Illustration of the decomposition used in our proof. In addition to the separable contribution $\Gamma(k_0)$, the superexponentially decaying remainder terms $\Delta_k$ are added to the identity contribution $\gamma(k)\id_{AC}$ (not depicted), leaving it separable.}
\end{figure}
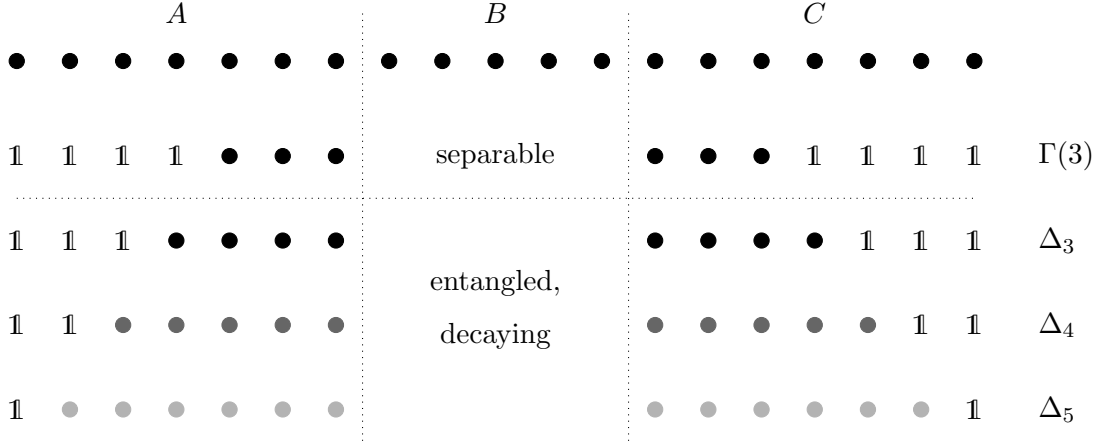
Note that the above summation is in fact finite as the terms are zero whenever $k\ge\max\{|A|,|C|\}$.
Let us first consider the terms in the sum separately.
\begin{equation}\label{eq:defDelta}
\begin{split}
        \Delta_k & :=\tr_B\Big[\rho^B E^{k+1\dagger}_{A,B}(1/2)E^{k+1\dagger}_{AB,C}(1/2)E^{k+1}_{AB,C}(1/2)E^{k+1}_{A,B}(1/2)\Big]\\
&\phantom{=====}-\tr_B\Big[\rho^B E^{k\dagger}_{A,B}(1/2)E^{k\dagger}_{AB,C}(1/2)E^{k}_{AB,C}(1/2)E^{k}_{A,B}(1/2)\Big]
\end{split}
\end{equation}
We note that $\Delta_k$ is supported on $\partial^{k+1} B\cap (A\cup C)$, which is of size $2k+2$.
Further we can bound the operator norm of $\Delta_k$:
\begin{align*}
        \|\Delta_k\|&\le\Big\|E^{k+1\dagger}_{A,B}(1/2)E^{k+1\dagger}_{AB,C}(1/2)E^{k+1}_{AB,C}(1/2)E^{k+1}_{A,B}(1/2)\\
                    &\phantom{==}-E^{k\dagger}_{A,B}(1/2)E^{k\dagger}_{AB,C}(1/2)E^{k}_{AB,C}(1/2)E^{k}_{A,B}(1/2)\Big\|\\
                    &\le \left\|E^{k+1\dagger}_{A,B}(1/2)-E^{k\dagger}_{A,B}(1/2)\right\|\left\|E^{k+1\dagger}_{AB,C}(1/2)E^{k+1}_{AB,C}(1/2)E^{k+1}_{A,B}(1/2)\right\|\\
                    &\phantom{==}+\left\|E^{k\dagger}_{A,B}(1/2)\right\|\left\|E^{k+1\dagger}_{AB,C}(1/2)-E^{k\dagger}_{AB,C}(1/2)\right\|\left\|E^{k+1}_{AB,C}(1/2)E^{k+1}_{A,B}(1/2)\right\|\\
                    &\phantom{==}+\left\|E^{k\dagger}_{A,B}(1/2)E^{k\dagger}_{AB,C}(1/2)\right\|\left\| E^{k+1}_{AB,C}(1/2)-E^{k}_{AB,C}(1/2)\right\|\left\|E^{k+1}_{A,B}(1/2)\right\|\\
                    &\phantom{==}+\left\|E^{k\dagger}_{A,B}(1/2)E^{k\dagger}_{AB,C}(1/2)E^{k}_{AB,C}(1/2)\right\|\left\|E^{k+1}_{A,B}(1/2)-E^{k}_{A,B}(1/2)\right\|\\
                    &\le4\cG^3 \frac{\cG^k}{(\lfloor k/r\rfloor+1)!}\,,
\end{align*}
where the first inequality follows from Lemma~\ref{lem:pTrContract}, the second from the triangle inequality and submultiplicativity of the operator norm, and the last line is an application of Lemma~\ref{lem:arakiExp}.

This however does not quite suffice to apply Lemma~\ref{lem:sepPerturb} yet, as $\Delta_k$ is a difference of operators and may not be bounded away from zero.

Let us therefore, consider the first expression in Line~\eqref{eq:k0region}, which is equal to
\begin{equation}\label{eq:k0SeparDecomp}
        \frac{Z_{k_0}}{Z_B}\exp(H^{k_0}_{AC}/2)\rho_{AC}^{k_0}\exp(H^{k_0}_{AC}/2)=\frac{Z_{k_0}}{Z_B}\Gamma(k_0)+\frac{Z_{k_0}}{Z_B}\gamma(k_0)\id_{AC}
\end{equation}
by Proposition~\ref{prop:constACsep} as long as $|B|\ge \ell_1(k_0)$. While both terms are already separable, we will view the $\Delta_k$ as a perturbation to the second term that remains separable.

Noting again that by Lemma~\ref{lem:partFuncRatio}, we have
\begin{equation}\label{eq:Zratk0}
        \frac{Z_{k_0}}{Z_B}\ge\exp(-\left\|H^{k_0}_{ABC}-H_B\right\|+2k_0\log(d))\ge \exp(-2k_0J+2k_0\log(d))
\end{equation}
Hence, we can write 
\begin{equation}\label{eq:idLow}
        \frac{Z_{k_0}}{Z_B}\gamma(k_0)\ge C \exp(-\alpha k_0)
\end{equation}
for some constants $C,\alpha>0$ that combine the ones from Equation~\eqref{eq:Zratk0} with the ones lower bounding $\gamma(k)$ in Proposition~\ref{prop:constACsep}.

We now decompose
\begin{align*}
        C\exp(-\alpha k_0)\id_{AC}+\sum_{k=k_0}^\infty\Delta_k=\sum_{k=k_0}^\infty C e^{-\alpha k_0-\log(2)(k-k_0+1)}\id_{AC}+\Delta_k.
\end{align*}
For this to be separable by Lemma~\ref{lem:sepPerturb}, we need to ensure that
\[
        \|\Delta_k\|\le d^{-k-1} \frac1Ce^{-\alpha k_0-\log(2)(k-k_0+1)}\,,
\]
i.e.,
\[
        \cG^3\frac{\cG^k}{(\lfloor k/r\rfloor+1)!}\le d^{-k-1} \frac1Ce^{-\alpha k_0-\log(2)(k-k_0+1)}\,,
\]
for all $k\ge k_0$.
This condition can be summarized as
\[
        C' \exp(\alpha' k+\alpha_0 k_0)\le (\lfloor k/r\rfloor+1)!
\]
for some new constants $C'\ge1$, $\alpha'>0$, $\alpha_0$ and is required $\forall k\ge k_0$.
By Robbin-Stirling's formula $n!\ge(n/e)^n$, this is satisfied if we enforce the (rather loose) bound\footnote{Choosing $k_0\ge r\exp(r\alpha'+1)$, the right-hand side of $C'\exp(\alpha'k+\alpha_0k_0)\le\exp(k(\log(k/r)-1)/r)$ grows faster than the left-hand side in $k$ so it suffices to enforce it for $k=k_0$. Replacing $C'$ by $C'^{k_0}$ strengthens the inequality and makes it solvable.}
\[
        k_0:=re\exp(r(\log(C')+\alpha_0+\alpha'))
\]
and so we set our entanglement lengthscale $\ell(J,d,r)=\ell_1(k_0)$, where the dependency on the parameters enters through the choices of constants and the implicit dependence of the function $\ell_1$ from Proposition~\ref{prop:constACsep}.

Let us now summarize the decomposition to show that this closes the proof:
We combine \eqref{eq:k0SeparDecomp} and \eqref{eq:defDelta} into \eqref{eq:sandwMarg}-\eqref{eq:decompLast}
\begin{align*}
        \frac{Z_{ABC}}{Z_B}e^{H_{AC}/2}\rho_{AC}e^{H_{AC}/2}&=\frac{Z_{k_0}}{Z_B}\left(\Gamma(k_0)+\gamma(k_0)\id_{AC}\right)+\sum_{k=k_0}\Delta_k\\
                                                      &=\frac{Z_{k_0}}{Z_B}\Gamma(k_0)+\left(\frac{Z_{k_0}}{Z_B}\gamma(k_0)-C e^{-\alpha k_0}\right)\id_{AC}\\
                                                      &\phantom{=}+\sum_{k=k_0}^\infty C e^{-\alpha k_0-\log(2)(k-k_0+1)}\id_{AC}+\Delta_k\,,
\end{align*}
where the first term was already chosen separable by Proposition~\ref{prop:constACsep}, the second term is positive (and thereby separable) by Equation~\eqref{eq:idLow}, and separability of the last term was just shown for the given entanglement lengthscale.
\end{proof}

\section{The thermodynamic limit}\label{sec:KMS}
In this section, we discuss the implications of our result for the thermodynamic limit by formulating it in terms of states on the quasilocal algebra of the infinite systems.
This is a straightforward corollary as our statement already holds uniformly in system size, but requires some additional setup and definitions.
We recall the quasilocal algebra $\cA=\overline{\cup_n\cA_{[-n,n]}}$
and in addition define the left and right algebras $\cA_L=\overline{\cup_n\cA_{[-n,0]}}$, $\cA_R=\overline{\cup_n\cA_{[1,n]}}$, where all closures are with respect to the operator norm.
Note that the quasilocal algebra is \emph{not} the algebraic tensor product of the left and right algebras.
However, evidently we can still embed the left and right algebra in the quasilocal algebra $\cA_L,\cA_R\subset\cA$ and $\cA_L\cap\cA_R=\CC\id$.

The following definition is taken from \cite[Definition 3.2]{keyl2006}.
\begin{definition}
        A state $\omega$ on $\cA$ is called a product state between $\cA_L$ and $\cA_R$ if for all $A\in\cA_L$, $B\in\cA_R$ it satisfies $\omega(AB)=\omega(A)\omega(B)$. A state is called separable if it lies in the weak*-closure of the convex hull of product states.
\end{definition}
Further, recall that our model of bounded finite-range interaction defines a time evolution on $\cA$ as the automorphism group $\alpha_t=\lim_{n\to\infty}\exp(it[H_{[-n,n]},\cdot])$ and that a state $\omega$ is a KMS state $\omega$ (at $\beta=1$) if
\[
        \omega(A\alpha_i(B))=\omega(BA)
\]
for all $A$, $B$ in a norm-dense, $\alpha_t$-invariant *-subalgebra of $\cA$ \cite[Definition 5.3.1]{bratteli1987a}, but we will not work directly with the definition.

We still denote by $\tr_{[-n,0]}:\cA_L\to\cA_L$ the partial trace extended by continuity to the closure of the algebra. Note that according to the usual embedding we identify $\tr_{[-n,0]}[A]=\tr_{[-n,0]}[A]\otimes \id_{[-n,0]}$.

\begin{corollary}\label{cor:KMS}
        Let $\Phi$ be an interaction of range $r$ and interaction strength $J$ and let $\omega$ be the KMS state.
        There exists an $\ell\in\NN$ only dependent on $r$, $J$, and $d$, such that
        \[
                \omega\circ\tr_{[-\ell,0]}
        \]
        is separable.
\end{corollary}

\begin{proof}
        Consider the density matrices $\rho^{[-m,m]}$, which define states on $\cA_{[-m,m]}=\cA_{[-m,0]}\otimes\cA_{[1,m]}$.
        The typical construction of thermal states extends these states using the Hahn-Banach theorem to the quasilocal algebra and then proceeds to prove that the weak-* limit in $\Lambda$ is a KMS state.
        However, since the first part of this construction, the extension by Hahn-Banach, may introduce additional correlations/entanglement between the left and right subsystems, we choose a more explicit construction as follows:
        Define the states $\omega^m$ on $\cA_\Lambda$ for any $\Lambda\supset[-m,m]$ by
        \[
                \omega^m(A)=\Tr_\Lambda\left[\left(\rho^{[-m,m]}\otimes\frac{\id_{\Lambda\setminus[-m,m]}}{d^{2m+1}}\right)A\right]
        \]
        Since these are consistent with the embedding into larger algebras this defines a state on $\cA_0$, which we also denote by $\omega^m$.
        For $A\in\cA$, we choose $A_n\in\cA_0$ with $A_n\to A$ and define $\omega^m(A)=\lim_{n\to\infty}\omega^m(A_n)$, which is unique since for $\cA_0\ni B_n\to A$
        \[
                \lim_{n\to\infty}\left|\omega^m(A_n)-\omega^m(B_n)\right|\le\lim_{n\to\infty} \|A_n-B_n\|=0\,.
        \]

        First note that the sequence $\omega^m$ has a weak-* convergent subnet $\omega^\alpha$, and that every limit point is a KMS state \cite[Proposition 6.2.15]{bratteli1987a}, which by \cite{araki1975} is unique in our setup of one-dimensional finite-range interactions, so all limit points coincide.
        We denote the KMS state by $\omega$ and observe that 
        \[
                \omega\circ\tr_{[-\ell,0]}=\left(\lim_\alpha\omega^\alpha\right)\circ\tr_{[-\ell,0]}=\lim_\alpha\left(\omega^\alpha\circ\tr_{[-\ell,0]}\right)
        \]
        as the limit is taken in the weak-* sense and $\tr_{[-\ell,0]}(\cA)\subset\cA$.

        Due to the closedness of the set of separable operators in the weak-* topology we are left to show that each $\omega^\alpha\circ\tr_{[-\ell,0]}$ (or each $\omega^m\circ\tr_{[-\ell,0]}$) is separable.
        By Theorem~\ref{thm:mainTec}, choosing a sufficiently large $\ell$, we can decompose
        \begin{equation}\label{eq:KMSdecompA0}
                \omega^m\circ\tr_{[-\ell,0]}=\sum_{k} p_k \omega^{m,k}_L\otimes\omega^{m,k}_R:=\sum_k p_k \omega^{m,k}
                \end{equation}
        on $\cA_0$, where the sum is finite for each $m$ by Carath\'eodory's theorem.
        Again we define $\omega^{m,k}$ on $\cA$ using limits of approximating sequences, which are unique due to boundedness as before.
        Since Equation~\eqref{eq:KMSdecompA0} holds on a dense subset of $\cA$, it remains true for their extensions to $\cA$.
        Finally, for each $\omega^{m,k}$ and for $A\in\cA_L$, $B\in\cA_R$, we can find an approximating sequence $A_n\to A$, $B_n\to B$ (recall the limit is independent of the choice), where $A_n\in\cA_0\cap\cA_L$, $B_n\in\cA_0\cap\cA_R$ and
        \begin{align*}
                \omega^{m,k}(AB)& =\lim_{n\to\infty}\omega^{m,k}(A_nB_n)\\
                                &= \lim_{n\to\infty}\omega^{m,k}(A_n)\omega^{m,k}(B_n)\\
                                &=\lim_{n\to\infty}\omega^{k,m}(A_n)\omega^{m,k}(B_n)\\
                                &=\omega^{m,k}(A)\omega^{m,k}(B)\,.
\end{align*}
The first equality  uses continuity of $\tr_{[-\ell,0]}$, the second equality the fact that $\omega^{m,k}$ is product in $\cA_0$, and the rest is definition.
So $\omega^{m,k}$ is product on $\cA$, $\omega^m\circ\tr_{[-\ell,0]}$ is separable, and the proof follows.

\end{proof}

\section*{Acknowledgements}
I would like to thank Isaac Kim, Bruno Nachtergaele, Andreas Bluhm, and Ainesh Bakshi for helpful discussions.
This material is based upon work supported by the U.S. Department of Energy, Office of Science, National Quantum Information Science Research Centers, Quantum Science Center.
\clearpage
\sloppy
\printbibliography
\end{document}